\newcommand{\CLASSINPUTbaselinestretch}{1.2}
\documentclass[conference,twocolumn]{IEEEtran}

\usepackage{geometry}
\usepackage[english]{babel}
\usepackage{enumitem}
\usepackage[latin1]{inputenc}
\usepackage{enumerate}
\usepackage{color}
\usepackage[T1]{fontenc}
\usepackage{subfigure}
\usepackage{dsfont}
\usepackage[final]{graphicx}
\usepackage[T1]{fontenc}
\usepackage{amsmath}
\usepackage{mathtools, cuted}
\usepackage{amsthm}
\usepackage{amstext}
\usepackage{amssymb}
\usepackage{mathrsfs}
\usepackage{cite}
\usepackage{mathtools}
\usepackage{tikz}
\usepackage{marginnote}

\usepackage{stackengine}
\usepackage{MNsymbol}
\newcommand*\mnote[3][0pt]{%
  \if l#2\reversemarginpar\def\pointer{\filledmedtriangleright}%
    \def\stackalignment{r}\fi%
  \if r#2\normalmarginpar\def\pointer{\filledmedtriangleleft}%
    \def\stackalignment{l}\fi%
  \marginpar{%
    \topinset{%
      \scalebox{1.5}{\textcolor{blue}{$\pointer$}}}{%
      \belowbaseline[-1.5\baselineskip-#1]{%
        \stackengine%
          {-5pt}%
          {\fcolorbox{blue}{white}{\parbox{1.8cm}%
            {\vspace{3pt}\raggedright#3}}}%
          {~\colorbox{white}{\sffamily \textcolor{red}{Note}}}%
          {O}%
          {l}%
          {F}%
          {F}%
          {S}%
        }%
      }{%
      3ex+#1}{%
      -2ex}%
  }%
}

\usetikzlibrary{arrows,positioning, shapes}
\usepackage{float}

\usepackage{setspace}
\singlespace

\definecolor{Darkgreen}{rgb}{0,0.4,0}
\definecolor{wine-stain}{rgb}{0.5,0,0}
\usepackage{hyperref}
\hypersetup{colorlinks,
    colorlinks,
    citecolor=red,
    filecolor=green,
    linkcolor=blue,
    linktoc=page,
    urlcolor=blue
}

\def\R{\mathbb{R}}

\def\eps{\varepsilon}

\def\B{{\mathcal B}}

\def\E{{\mathbb E}}

\def\C{{\mathcal C}}

\def\X{{\mathcal X}}
\def\Y{{\mathcal Y}}

\def\S{{\mathcal S}}

\def\Z{{\mathcal Z}}
\def\U{{\mathcal U}}
\def\V{{\mathcal V}}

\def\sBer{{\mathsf{Bernoulli}}}
\def\sP  {\mathsf{P}_{\eps}^{\mathsf{e}}}

\def\sM{{\mathsf {sENSR}}}

\def\sW{{\mathsf {wENSR}}}
\def\sG{{\mathsf G}}
\def \var {{\mathsf {var}   }}
\def \mmse {{\mathsf {mmse}   }}

\usepackage[font=small,labelsep=space]{caption}
\DeclareCaptionLabelSeparator{dot}{.~}
\newcounter{example}
\newenvironment{example}[1][]{\refstepcounter{example}\par\medskip
   \noindent \textit{Example~\theexample. #1} \rmfamily}{\medskip}
\newtheorem{definition}{Definition}
\newtheorem{theorem}{Theorem}
\newtheorem{corollary}{Corollary}
\newtheorem*{corollary-non}{Corollary}

\newtheorem{proposition}{Proposition}
\newtheorem{lemma}{Lemma}
\theoremstyle{remark}
\newtheorem{remark}{Remark}

\newcommand{\markov}{\mathrel\multimap\joinrel\mathrel-%
\mspace{-9mu}\joinrel\mathrel-}

\usepackage{times}
\usepackage{tikz}
\usepackage{amsmath}
\usepackage{verbatim}
\usetikzlibrary{arrows,shapes}
\tikzstyle{RectObject}=[rectangle,fill=white,draw,line width=0.2mm]
\tikzstyle{line}=[draw]
\tikzstyle{arrow}=[draw, -latex]
\usetikzlibrary{decorations.pathmorphing}
\usetikzlibrary{calc,shapes, positioning}
\usepackage{graphicx}
\usepackage{caption}
\usetikzlibrary{shapes.geometric}

\IEEEoverridecommandlockouts

\allowdisplaybreaks

\begin{document}

\title{\vspace{5.5mm}{Privacy-Aware MMSE Estimation}}
\author{\IEEEauthorblockN{Shahab Asoodeh, Fady Alajaji, and Tam\'{a}s Linder}
    \IEEEauthorblockA{\normalsize{Department of Mathematics and Statistics, Queen's University}
    \\\{asoodehshahab, fady, linder\}@mast.queensu.ca}}
\restoregeometry
\maketitle

\begin{abstract}
We investigate the problem of the predictability of random variable $Y$ under a privacy constraint dictated by random variable $X$, correlated with $Y$, where both predictability and privacy are assessed in terms of the minimum mean-squared error (MMSE). Given that $X$ and $Y$ are connected via a binary-input symmetric-output (BISO) channel, we derive the \emph{optimal} random mapping $P_{Z|Y}$ such that the MMSE of $Y$ given $Z$ is minimized while the MMSE of $X$ given $Z$ is greater than $(1-\eps)\var(X)$ for a given $\eps\geq 0$. We also consider the case where $(X,Y)$ are continuous and $P_{Z|Y}$ is restricted to be an additive noise channel.
\end{abstract}

\section{Introduction and Preliminaries}
Consider two communicating agents Alice and Bob. Alice observes a random variable $Y$ and wants to reveal it to Bob in order to receive a payoff. On the other hand, nature chooses $X$, dependent on $Y$ via a fixed channel $P_{X|Y}$. Alice wishes to disclose $Y$ as accurately as possible, but in such a way that $X$ is kept almost private from Bob. For instance, $Y$ may represent the information that a social network (Alice) obtains from its users and $X$ may represent political preferences of the users. Alice wants to disclose $Y$ as accurately as possible to an advertising company and, simultaneously, wishes to protect the privacy of its users.
Given a fixed joint distribution $P_{XY}$, Alice, hence, needs to choose a random mapping $P_{Z|Y}$, the so-called \emph{privacy filter}, to release a new random variable $Z$, called the \emph{displayed data}, such that $X$ and $Z$ satisfy a privacy constraint and $Z$ maximizes a utility function (corresponding to the predictability of $Y$).

This problem has been addressed from an information-theoretic viewpoint in \cite{yamamotoequivocationdistortion,Lalitha_Forensics,Asoode_submitted,Asoodeh_Allerton,Asoodeh_CWIT,Calmon_fundamental-Limit,t_closeness,Fawaz_Makhdoumi,Funnel} where both utility and privacy are measured in terms of information-theoretic quantities.
In particular, in \cite{Asoodeh_Allerton} \emph{non-trivial perfect privacy} for discrete $X$ and $Y$ where $Z$ is
required to be statistically independent of $X$ and dependent on $Y$, is studied. It is shown that non-trivial perfect
privacy is possible if and only if  $X$ is \emph{weakly independent} of $Y$, that is, if the set of
vectors $\{P_{X|Y}(\cdot):y\in \Y\}$ is linearly dependent. Calmon et al.\ \cite{Calmon_fundamental-Limit} showed that $X$ is weakly independent of $Y$ if and only if the smallest singular value of the conditional expectation operator $f\mapsto \E[f(X)|Y]$ is zero and hence obtained an equivalent necessary and sufficient condition of non-trivial perfect privacy.

 In this paper, we take an estimation-theoretic approach and define both the privacy and utility functions in terms of the minimum mean-squared error (MMSE). For a given pair of random variables $(U,V)$, the MMSE of estimating $U$ given $V$ is
 \begin{eqnarray*}
  \mmse(U|V) &:=& \inf_{g\in \B(\R)}\E[(U-g(V))^2]\\
  &=&\E[\left(U-\E[U|V]\right)^2]=\E[\var(U|V)],
\end{eqnarray*}
 where $\B(\R)$ denotes the collection of all Borel measurable\footnote{As pointed out in \cite{Counterexample_WISE}, we need to restrict the minimization to the collection of Borel measurable estimators $g$. It is possible to construct a nonmeasurable transformation $\hat{g}$ yielding a random variable $\hat{g}(V)$ which is equal to $U$ pointwise but $\mmse(U|V)=\var(U)>0$.} functions on the real line and $\var(\cdot|\cdot)$ denotes the conditional variance. The privacy filter $P_{Z|Y}$ is said to satisfy the $\eps$-\emph{strong estimation privacy} condition if $\mmse(f(X)|Y)\geq (1-\eps)\var(f(X))$ for any Borel function\footnote{This is reminiscent of \emph{semantic security} \cite{Goldwasser1984270} in the cryptography community. An encryption mechanism is said to be semantically secure if the adversary's advantage for correctly guessing \emph{any function} of the privata data given an observation of the mechanism's output (i.e., the ciphertext) is required to be negligible.} $f$ of $X$ and some $\eps\geq 0$ and similarly, it is said to satisfy the $\eps$-\emph{weak estimation privacy} condition if $\mmse(X|Y)\geq (1-\eps)\var(X)$. The parameter $\eps$ determines the level of desired privacy; in particular, $\eps=0$ corresponds to perfect privacy. We propose to use the estimation noise to signal ratio (ENSR), defined by $\frac{\mmse(Y|Z)}{\var(Y)}$, as the loss function associated with $Y$ and $Z$. The goal is to choose $P_{Z|Y}$ which satisfies the strong (resp., weak) estimation privacy condition and \emph{minimizes} the ENSR (or equivalently maximizes $\frac{\var(Y)}{\mmse(Y|Z)}$ as the utility function), which ensures the best predictability of $Y$ given a privacy-preserving $Z$. The function $\sM_\eps(X;Y)$  (resp., $\sW_\eps(X;Y)$) is introduced as this minimum to quantify the above goal.

  To evaluate $\sM_\eps(X;Y)$, we first show that the $\eps$-strong estimation privacy condition is equivalent to $\rho_m^2(X;Y)\leq \eps$ where $\rho_m$ is the maximal correlation.  We then show that $\sM_\eps(X;Y)$ and $\sW_\eps(X;Y)$ admit closed-form expressions when $P_{X|Y}$ is a binary-input and symmetric-output (BISO) channel.  Moreover, when $X$ is discrete, we develop a bound characterizing the privacy-constrained error probability, $\Pr(\hat{Y}(Z)\neq Y)$, for all estimators $\hat{Y}(Z)$ given a privacy-preserving $Z$, thus generalizing the results of \cite{Calmon_bounds_Inference}. In particular, we show that the fundamental bound on privacy-constrained error probability decreases \emph{linearly} as $\eps$ increases, analogously to \cite[Corollaries 3,5]{Calmon_bounds_Inference}.

We also study $\sM_\eps(X^n; Y^n)$ when $n$ i.i.d. copies $(X^n, Y^n)$ of $(X,Y)$ are available. It is intuitively clear from the Slepian-Wolf theorem that non-trivial perfect privacy is always possible for $(X^n,Y^n)$ with
sufficiently large $n$ irrespective of the perfect privacy associated with $(X,Y)$. 
This observation is formalized by Calmon et al.\ \cite{Calmon_fundamental-Limit} by showing that, unless $X$ is a deterministic function of $Y$, the smallest singular value of the operator $f(X^n)\mapsto \E[f(X^n|Y^n)]$ converges to zero as $n\to \infty$, and hence non-trivial perfect privacy is possible for sufficiently large $n$. However, we demonstrate that if the class of privacy filters is constrained to be memoryless, then the situation drastically changes and $\sM_\eps(X^n; Y^n)$ remains the same for any $n$. This is reminiscent of the tensorization
property for the maximal correlation proved in \cite{Witsenhausen:dependent}.

In addition, $\sM_\eps(X; Y)$ is considered for the case where $(X,Y)$ has a joint probability density function by studying the problem where the displayed data $Z$ is obtained by passing $Y$ through an additive-noise channel. In
this case, we show that for a Gaussian noise process, jointly Gaussian $(X_{\sG},Y_{\sG})$ is
the worst case (i.e., has the largest ENSR). We also show that if $Y_{\sG}$ is Gaussian then the ENSR of $(X,Y_{\sG})$ is very close to the Gaussian ENSR if the maximal correlation between $X$ and $Y_{\sG}$ is close to the
correlation coefficient between $X$ and $Y_{\sG}$. It is important to note that maximal correlation is weakly lower semi-continuous, and hence the fact that $\rho_m^2(X; Y_{\sG})$ is close to $\rho^2(X; Y_{\sG})$ does not necessary mean that $X$ is Gaussian.

The rest of this paper is organized as follows. In Section II, we formally formulate the problem in terms of the strong and weak estimation privacy conditions and obtain some equivalent formulations. In Section III, we focus on discrete $(X, Y)$ and derive some properties for the corresponding utility-privacy functions and then calculate $\sM_\eps(X;Y)$ and $\sW_\eps(X;Y)$ for binary $Y$. Section IV is devoted to the same problem for continuous $(X,Y)$ when the privacy filter is an additive-noise channel.

\section{Strong estimation privacy guarantee}
Consider the scenario where Alice observes $Y$ which is correlated with a private random variable $X$, drawn from a given joint distribution $P_{XY}$, and wishes to transmit the random variable $Z$ to Bob to receive some utility from him. Her goal is to maximize the utility while making sure that Bob cannot efficiently estimate any non-trivial function of $X$ given $Z$. To formalize this privacy guarantee, we give the following definition. In what follows random variables $X$, $Y$, and $Z$ have alphabets $\X$, $\Y$, and $\Z$, respectively, which are either finite subsets of $\R$ or they are all equal to $\R$.
\begin{definition}\label{Def:strong_estimation_privacy}
Given a joint distribution $P_{XY}$ and $\eps\geq0$, $Z$ is said to satisfy \emph{$\eps$-strong estimation privacy}, denoted as $Z\in \Gamma_{\eps}(P_{XY})$, if there exists a random mapping (channel) $P_{Z|Y}$ that induces a joint distribution $P_X\times P_{Z|X}$ on $\X\times \Z$, via the Markov condition $X\markov Y\markov Z$, satisfying
\begin{equation}\label{Eq:strong_estimation_privacy}
  \mmse(f(X)|Z)\geq (1-\eps)\var(f(X)),
\end{equation}
for any non-degenerate Borel functions $f$ on $\X$. Similarly, $Z$ is said to satisfy  \emph{$\eps$-weak estimation privacy}, denoted as $Z\in \partial\Gamma_{\eps}(P_{XY})$, if \eqref{Eq:strong_estimation_privacy} is satisfied only for the identity function $f(x)=x$.
\end{definition}
In the sequel, we drop in the notation the dependence of $\Gamma_{\eps}(P_{XY})$ on $P_{XY}$ and simply write $\Gamma_\eps$.

Suppose the utility Alice receives from Bob is $\frac{\var(Y)}{\mmse(Y|Z)}$. The utility is maximized (and is equal to $\infty$) when $Z=Y$ with probability one and is minimized (and is equal to one) when $Z$ is independent of $Y$. In order to quantify the tradeoff between privacy guarantee (introduced above) and the utility, we propose the following function, which we call the strong privacy-aware \emph{estimation noise to signal ratio} (ENSR):
\begin{equation}\label{Def:M_Eps}
  \sM_{\eps}(X;Y):=\inf_{Z\in \Gamma_{\eps}}\frac{\mmse(Y|Z)}{\var(Y)}.
\end{equation}
Similarly, we can use weak estimation privacy to define the weak privacy-aware ENSR as follows:
\begin{equation}\label{Def:W_Eps}
\sW_{\eps}(X;Y):=\inf_{Z\in \partial\Gamma_{\eps}}\frac{\mmse(Y|Z)}{\var(Y)}.
\end{equation}
\begin{remark}\label{remark_Correlation_Ratio}
  The quantity $\displaystyle \frac{\mmse(Y|Z)}{\var(Y)}$ is intimately related to the \emph{correlation ratio}, introduced by R\'{e}nyi \cite{Renyi-dependence-measure}. The correlation ratio of $Y$ on $Z$, denoted by $\eta_Z(Y)$, is defined as $$\eta^2_Z(Y):=\frac{\var(\E[Y|Z])}{\var(Y)},$$ which can be shown to be equal to $\sup_{g} \rho^2(Y; g(Z))$, where $\rho$ is the standard correlation coefficient. It is clear from the law of total variance that $$ \frac{\mmse(Y|Z)}{\var(Y)}=1-\eta^2_Z(Y).$$
\end{remark}
In the sequel, we obtain an equivalent characterization for the random mapping $P_{Z|X}$ which generate $Z\in \Gamma_{\eps}$. To this goal, we need the following definition.
\begin{definition}[\hspace{-0.007cm}\cite{Sarmanov,Renyi-dependence-measure}]\label{Definition-Maximal_corr}
Given  random  variables $U$ and $V$ taking values over arbitrary alphabets $\U$ and $\V$, respectively, the  \emph{maximal  correlation} $\rho_m(U;V)$ is defined as
\begin{eqnarray*}
  \rho_m^2(U;V)&:=& \sup_{f, g} \rho^2(f(U), g(V))\\
   &=& \sup_{(f(U),g(V))\in \mathcal{S}^0}\frac{\E^2[f(U)g(V)]}{\var(f(U))\var(g(V))},
\end{eqnarray*}
where $\mathcal{S}^0$ is the collection of all pairs of real-valued measurable functions $f$ and $g$ of $U$ and $V$, respectively, such that $\E[f(U)]=\E[g(V)]=0$ and $0<~\var(f(U)), \var(g(V))<~\infty$.
\end{definition}
It can be shown that $0\leq\rho_m(U; V)\leq 1$ where the lower bound is achieved if and only if $U$ and $V$ are independent and the upper bound is achieved if and only if there exists a pair of functions $(f, g)\in \S^0$ such that $f(U)=g(V)$ almost surely. R\'{e}nyi \cite{Renyi-dependence-measure} derived an equivalent characterization of maximal correlation as
\begin{equation}\label{Maximal_correlation_Equivalent}
  \rho^2_m(U;V)=\sup_{f\in\S^0_{\U}}\frac{\E\left[\E^2[f(U)|V]\right]}{\var(f(U))},
\end{equation}
where $\mathcal{S}_{\U}^0$ is the collection of all real-valued measurable functions $f$ of $U$ such that $\E[f(U)]=0$ and $0<\var(f(U))<\infty$.
\begin{theorem}\label{Theorem_equivakebt_strong_MC}
     For a given $P_{XY}$, $Z\in \Gamma_{\eps}$ if and only if there exists $P_{Z|Y}$ which induces $P_{Z|X}$ via $X\markov Y\markov Z$ satisfying $\rho_m^2(X;Z)\leq \eps$ for any $\eps\geq 0$.
\end{theorem}
\begin{proof}
Consider a function $f:\X\to \R$. We can define $\tilde{f}(X):=f(X)-\E[f(X)]$ and since $\mmse(\tilde{f}(X)|Z)=\mmse(f(X)|Z)$ and $\var(\tilde{f}(X))=\var(f(X))$, without loss of generality, we can assume that $\E[f(X)]=0$. We can then write
\begin{equation}\label{Eq:Proof_lemma1}
\eta^2_Z(f(X))=\frac{\E[\E^2[f(X)|Z]]}{\var(f(X))}  
\end{equation}
Thus we obtain
\begin{eqnarray}
  \inf_{f\in\S_{\X}^0}\frac{\mmse(f(X)|Z)}{\var(f(X))} &=& 1-\sup_{f\in\S_{\X}^0}\eta^2_Z(f(X))\label{Eq:Proof_lemma1_3}\\
   &\stackrel{(a)}{=}& 1-\rho_m^2(X;Z),\label{Eq:Proof_lemma1_2}
\end{eqnarray}
where \eqref{Eq:Proof_lemma1_2} is due to \eqref{Maximal_correlation_Equivalent}.

If $\rho_m^2(X;Z)\leq \eps$, then it is clear from \eqref{Eq:Proof_lemma1_2} that
\begin{equation*}
   \mmse(f(X)|Z)\geq (1-\eps)\var(f(X))
\end{equation*}
and hence \eqref{Eq:strong_estimation_privacy} is satisfied. Conversely, let $P_{XZ}$ satisfy the $\eps$-strong estimation privacy. Then for any $f$, \eqref{Eq:strong_estimation_privacy} is satisfied. Also, in view of \eqref{Eq:Proof_lemma1_3} and \eqref{Eq:Proof_lemma1_2} for arbitrary $\delta>0$, there exists $f\in \S^0_{\X}$ such that
$$1-\eps\leq \frac{\mmse(f(X)|Z)}{\var(f(X))}\leq 1-\rho_m^2(X;Z)+\delta,$$ and hence,
$$\rho_m^2(X;Z)\leq \eps+\delta,$$ which completes the proof.
\end{proof}
In light of Theorem~\ref{Theorem_equivakebt_strong_MC} and Remark~\ref{remark_Correlation_Ratio}, we can write $\sM_{\eps}(X;Z)$ and $\sW_{\eps}(X;Z)$ alternatively as
\begin{equation}\label{Eq:PAMMSE_equivalent}
  \sM_{\eps}(X;Y)=1-\sup_{P_{Z|Y}:\rho_m^2(X;Z)\leq \eps,\atop X\markov Y\markov Z} \eta^2_Z(Y),
\end{equation}
and
\begin{equation}\label{Eq:PAMMSE_equivalent_2}
  \sM_{\eps}(X;Y)=1-\sup_{P_{Z|Y}:\eta^2_Z(X)\leq \eps,\atop X\markov Y\markov Z} \eta^2_Z(Y),
\end{equation}
for any $\eps\geq 0$. We note that, using the Support Lemma \cite{csiszarbook}, one can show the set $\Gamma_{\eps}$ can be described only by considering $Z\in \Z$ with $|\Z|\leq |\Y|+1$ in case $\Y$ is finite. We also note that since both maximal correlation and correlation ratio satisfy the data processing inequality \cite{Asoode_submitted,Calmon_bounds_Inference,Ulukus_Data_processing}, i.e. $\rho_m^2(X; Z)\leq \eta_m^2(X;Y)$ and $\eta^2_Z(X)\leq \eta^2_Y(X)$ over $X\markov Y\markov Z$, we can restrict our attention to $0\leq\eps\leq \rho_m^2(X;Y)$ and $0\leq\eps\leq \eta_Y^2(X)$ in \eqref{Eq:PAMMSE_equivalent} and \eqref{Eq:PAMMSE_equivalent_2}, respectively.

\section{Characterization of  $\sM_{\eps}(X;Y)$ and $\sW_{\eps}(X;Y)$ For Discrete $X$ and $Y$}
We first derive some properties of $\sM_\eps(X;Y)$ and $\sW_\eps(X;Y)$ when both $X$ and $Y$ are discrete. For a given $P_{XY}$ and $0\leq\eps\leq\rho_m^2(X;Y)$, we have the following trivial bounds:
\begin{equation}\label{Bounds_UP_LB_Weak}
  0\leq\sW_{\eps}(X;Y)\leq \sM_{\eps}(X;Y)\leq 1-\eps,
\end{equation}
where the last inequality can be proved by noticing that $\sM_{\eps}(X;Y)\leq \sM_{\eps}(Y;Y)$ and
 \begin{eqnarray}
  \mmse(Y|Z)&=& \var(Y)(1-\eta^2_Z(Y))\nonumber\\ 
  &\geq & \var(Y)(1-\rho_m^2(Y;Z))\label{Inequality_MMSE_rho}, \end{eqnarray}
where \eqref{Inequality_MMSE_rho} follows from the definition of maximal correlation. The lower bound $0\leq\sM_{\eps}(X;Y)$ in \eqref{Bounds_UP_LB_Weak}  is achieved if and only if $\rho_m^2(X;Y)=\eps$. This is because $\sM_{\eps}(X;Y)=0$ implies that there exists $Z\in \Gamma_{\eps}$ such that $X\markov Y\markov Z$ and $\mmse(Y|Z)=0$ and hence $Z=Y$ almost surely and thus $Y\in \Gamma_{\eps}$. On the other hand, when $\eps=0$, the upper bound $\sM_{0}(X;Y)\leq 1$ is tight if and only if all $Z\in \Gamma_{0}$ are independent of $Y$. Hence, from \cite[Lemma 6]{Asoode_submitted}, $\sM_0(X;Y)=1$ if and only if $X$ is not \emph{weakly independent} of $Y$. In particular, if $|\Y|>|\X|$, then $\sM_{0}(X;Y)<1$, and if $|\Y|=2$, then $\sM_0(X;Y)=1$.

The map $\eps\mapsto \sM_{\eps}(X;Y)$ is clearly non-increasing. The following lemma states that this map is indeed convex and thus strictly decreasing. As another consequence of this convexity, we obtain an upper bound on $\sM_\eps(X;Y)$ which strictly strengthens \eqref{Bounds_UP_LB_Weak}.

\begin{lemma}\label{lemma_Concavity_M_eps}
For any joint distribution $P_{XY}$, the maps $\eps~\mapsto~\sM_{\eps}(X;Y)$ and $\eps~\mapsto~\sW_{\eps}(X;Y)$ are convex.
\end{lemma}

\begin{proof}
Here we give the complete proof for only $\sM_{\eps}(X;Y)$. The proof for $\sW_{\eps}(X;Y)$ is similar and hence is omitted.
%
 For brevity, in this proof we write $\sM_\eps$ instead
of $\sM_\eps(X;Y)$.
 It suffices to show that for any $0\leq \eps_1<\eps_2<\eps_3\leq \rho_m^2(X;Y)$, we have
\begin{equation}\label{concavity_requirement}
    \frac{\sM_{\eps_3}-\sM_{\eps_1}}{\eps_3-\eps_1}\geq \frac{\sM_{\eps_2}-\sM_{\eps_1}}{\eps_2-\eps_1},
\end{equation}
which, in turn, is equivalent to
\begin{equation}\label{concavity_requirement_2}
\sM_{\eps_2}\leq\left(\frac{\eps_2-\eps_1}{\eps_3-\eps_1}\right)\sM_{\eps_3}+\left(\frac{\eps_3-\eps_2}{\eps_3-\eps_1}\right)\sM_{\eps_1}.\end{equation}
Let $P_{Z_1|Y}: Y\to Z_1$   and $P_{Z_3|Y}:Y\to Z_3$ be two optimal channels with $Z_1\in \Gamma_{\eps_1}$, $Z_3\in \Gamma_{\eps_3}$, and  with disjoint output alphabets $\Z_1$ and $\Z_3$, respectively.

We introduce an auxiliary binary random variable $U\sim~\sBer(\lambda)$, independent of $(X,Y)$, where $\lambda:=\frac{\eps_2-\eps_1}{\eps_3-\eps_1}$ and define the channel $P_{Z_{\lambda}|Y}$: We pick $P_{Z_3|Y}$ if $U=1$ and $P_{Z_1|Y}$ if $U=0$, and let $Z_{\lambda}$ be the output of this channel with output alphabet $\Z_1\cup \Z_3$. We then have
\begin{eqnarray}
\E[\E^2[f(X)|Z_{\lambda}]] &=&\E\left[\E[\E^2[f(X)|Z_{\lambda}]|U]\right]\nonumber\\
&=& \lambda \E[\E^2[f(X)|Z_3]]\nonumber\\
&&+(1-\lambda) \E[\E^2[f(X)|Z_1]],\label{proof_Concavity_gHat}
\end{eqnarray}
where the second equality holds since $U$ is independent of $X$. We can then use the alternative characterization of maximal correlation in \eqref{Maximal_correlation_Equivalent} to write
\begin{eqnarray*}
  \rho^2_m(X;Z_{\lambda})&=&\sup_{f\in\S_{\X}^0}\frac{\E[\E^2[f(X)|Z_{\lambda}]]}{\E[f^2(X)]}  \\
   &\leq& \lambda \rho^2_m(X;Z_3)+(1-\lambda) \rho^2_m(X;Z_1)\\&\leq& \lambda \eps_3+(1-\lambda)\eps_1=\eps_2,
\end{eqnarray*}
where the first inequality follows from \eqref{proof_Concavity_gHat}. Thus $Z_{\lambda}~\in~\Gamma_{\eps_2}$.

On the other hand,  we have
\begin{eqnarray*}
 \mmse(Y|Z_{\lambda})&=&\E[Y^2]-\E[\E^2[Y|Z_{\lambda}]] \\
   &=&\E[Y^2]-\E[\E[\E^2[Y|Z_{\lambda}|U]]]   \\
    &=& \lambda \mmse(Y|Z_3)+(1-\lambda) \mmse(Y|Z_1),
\end{eqnarray*}
and hence

\begin{eqnarray*}
\sM_{\eps_2}&\leq& \frac{\mmse(Y|Z_{\lambda})}{\var(Y)}\\
&=&\frac{\lambda \mmse(Y|Z_3)-(1-\lambda) \mmse(Y|Z_1)}{\var(Y)}\\
&=& \lambda \sM_{\eps_3}+(1-\lambda) \sM_{\eps_1}
\end{eqnarray*}
which, according to \eqref{concavity_requirement_2}, completes the proof.
\end{proof}
In light of the convexity of $\eps\mapsto \sM_{\eps}(X;Y)$ the following corollaries are immediate.

\begin{corollary}\label{corollary_M_over_eps_Decreasing}
  For a given $P_{XY}$, the maps $\eps\mapsto \frac{1-\sM_{\eps}(X;Y)}{\eps}$ and $\eps\mapsto \frac{1-\sW_{\eps}(X;Y)}{\eps}$ are non-increasing over $(0, 1)$.
\end{corollary}
\begin{proof}
Consoider the map $\eps\mapsto \sM_0(X;Y)-\sM_{\eps}(X;Y)$. In view of Lemma~\ref{lemma_Concavity_M_eps}, this map is concave and consequently the chordal slope $\frac{\sM_0(X;Y)-\sM_{\eps}(X;Y)}{\eps}$ is decreasing in $\eps$. It therefore follows that
\begin{eqnarray*}
  \frac{1-\sM_{\eps}(X;Y)}{\eps} &=& \frac{1-\sM_{0}(X;Y)}{\eps} \\
   && +\frac{\sM_{0}(X;Y)-\sM_{\eps}(X;Y)}{\eps},
\end{eqnarray*}
is decreasing. The proof for $\sW_\eps(X;Y)$ follows similarly.
%
\end{proof}

\begin{corollary}\label{Corollary_UpperBound_M_eps}
For a given $P_{XY}$,
$$\sM_\eps(X;Y)\leq 1-\frac{1}{\rho_m^2(X;Y)}\min\{\eps, \rho_m^2(X;Y)\},$$
and
$$\sW_\eps(X;Y)\leq 1-\frac{1}{\eta_Y^2(X)}\min\{\eps, \eta_Y^2(X)\}.$$
\end{corollary}
\begin{proof}
Since $\eps\mapsto \sM_{\eps}(X;Y)$ is convex, it is always below the chord connecting $(0, \sM_0(X;Y))$ and $(\rho_m^2(X;Y), 0)$, and hence
$$\sM_\eps(X;Y)\leq \sM_0(X;Y)\left(1-\frac{\eps}{\rho_m^2(X;Y)}\right),$$ from which the result follows because $\sM_0(X;Y)\leq 1$. The proof for $\sW_{\eps}(X;Y)$ is similar.
\end{proof}

\begin{remark}\label{remark_Erasure_Optimal_Corollary_2}
  Note that simple calculations reveal that the upper bounds in Corollary~\ref{Corollary_UpperBound_M_eps} are achieved by an erasure channel (see Fig.~\ref{fig:LowerBound_Filter}). For example, the erasure channel that achieves the upper bound of $\sM_\eps(X;Y)$ is
$$P_{Z|Y}(z|y)=\left\{
  \begin{array}{ll}
    1-\tilde{\delta} , & \hbox{ \text{if} $z=y$} \\
    \tilde{\delta}, & \hbox{ \text{if} $z=$~\text{e},}
  \end{array}\right.$$
for all $y\in \Y$ and the erasure probability 
\begin{equation}\label{Delta_Tilde}
 \tilde{\delta}=1-\frac{\eps}{\rho_m^2(X;Y)},
\end{equation}
for $0\leq\eps\leq \rho_m^2(X;Y)$. This is because for the channel $P_{Z_{\delta}|Y}$, illustrated in Fig.~\ref{fig:LowerBound_Filter}, we have $\rho_m^2(X;Z_\delta)=(1-\delta)\rho_m^2(X;Y)$ and $\rho_m^2(Y;Z_\delta)=1-\delta$. Therefore, if $\delta=\tilde{\delta}$, defined in \eqref{Delta_Tilde}, $Z_{\tilde{\delta}}\in\Gamma_{\eps}$. A simple calculation verifies that for this channel $$\mmse(Y|Z_{\tilde{\delta}})=\var(Y)\tilde{\delta}=\var(Y)\left(1-\frac{\eps}{\rho_m^2(X;Y)}\right).$$

\end{remark}
\begin{figure}[t]
\centering
\begin{tikzpicture}
        \node (x) [circle] at (-3.7,-0.9) {$X$};
        \draw (-2,-1) node[fill=blue!20, anchor=base] (channel) {$~~P_{Y|X}~~$};
        \path [arrow] (x) -- (channel);
        \node (y) [circle] at (0,-0.9) {$Y$};
        \path [arrow] (channel) -- (y);
        \node (y1) [circle] at (0.15,-0.2) {};
        \node (y2) [circle] at (0.15,-0.6) {};
        \node (y3) [circle] at (0.15,-1) {};
        \node (y4) [circle] at (0.15,-1.4) {};
        \node (y5) [circle] at (0.15,-1.8) {};

        \node (z1) [circle] at (2,-0.2) {};
        \node (z2) [circle] at (2,-0.6) {};
        \node (z3) [circle] at (2,-1) {};
        \node (z4) [circle] at (2,-1.4) {};
        \node (z5) [circle] at (2,-1.8) {};
        \node (z6) [circle] at (2,-2.4) {};
        \node (z8) [circle] at (2.1,-2.35) {e};
        \node (z7) [circle] at (2.6,-0.9) {$Z_{\delta}$};
        \draw[thick]  (0.2,-0.2) -- (z1);
        \draw[thick]  (0.2,-0.2) -- (1.8,-2.4);
        \draw[thick]  (0.2,-0.6) -- (z2);
        \draw[thick]  (0.2,-0.6) -- (1.8,-2.4);
        \draw[thick]  (0.2,-1) -- (z3);
        \draw[thick]  (0.2,-1) -- (1.8,-2.4);
        \draw[thick]  (0.2,-1.4) -- (z4);
        \draw[thick]  (0.2,-1.4) -- (1.8,-2.4);
        \draw[thick]  (0.2,-1.8) -- (z5);
        \draw[thick]  (0.2,-1.8) -- (1.8,-2.4);
\end{tikzpicture}
\caption{\small{The channel that achieves the upper bound in Corollary~\ref{Corollary_UpperBound_M_eps} where $Z_{\delta}$ is the output of an erasure channel with erasure probability specified in \eqref{Delta_Tilde}.}} \label{fig:LowerBound_Filter}
\end{figure}
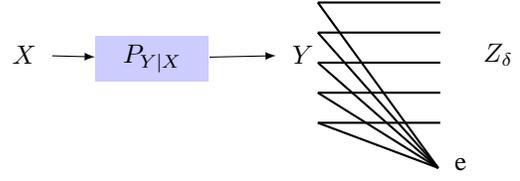

\subsection{Binary Input Symmetric Output $P_{X|Y}$}
We now turn our attention to the special case where $P_{X|Y}$ belongs to a family of channels called binary-input symmetric-output (BISO) channels, see e.g., \cite{Shamai_BISO1, Shamai_BISO2}.
For $Y\sim\sBer(p)$, $P_{X|Y}$ is BISO if, for any $x\in \X=\{0, \pm 1, \pm 2, \dots, \pm k\}$, we have $P_{X|Y}(x|1)=P_{X|Y}(-x|0)$. This clearly implies that $p_0:=P_{X|Y}(0|0)=P_{X|Y}(0|1)$. As pointed out in  \cite{Shamai_BISO2}, one can always assume that the output alphabet $\X=\{\pm 1, \pm 2, \dots, \pm k\}$ has even number of elements by splitting the symbol $0$ into two symbols and assigning equal probabilities.
This family of channels can also be characterized using the definition of \emph{quasi-symmetric} channels \cite[Definition 4.17]{Fady_Lecture_note}. A channel $\mathsf{W}$ is BISO if (after making $|\X|$ even) the transition matrix $P_{X|Y}$ can be  partitioned along its columns into binary-input binary-output sub-arrays in which rows are permutations of each other and the column sums are equal.
For example, binary symmetric channels and binary erasure channels are both BISO.

In the following theorem, we show that $\sW_\eps(X;Y)$ can be calculated in closed-form when $P_{X|Y}$ is a BISO channel.
\begin{theorem}\label{Lemma_BISO}
Let $Y\sim\sBer(p)$ and $P_{X|Y}$ be a BISO channel. Then for $0\leq\eps\leq \rho_m^2(X;Y)$, we have $$\sW_\eps(X;Y)=1-\eps\frac{\var(X)}{4\var(Y)\E^2[X|Y=1]},$$
and
$$1-\eps\frac{\var(X)}{4\var(Y)\E^2[X|Y=1]}\leq \sM_\eps(X;Y)\leq  1-\frac{\eps}{\rho_m^2(X;Y)}.$$
\end{theorem}
\begin{proof}
The proof is given in Appendix~\ref{Appebdix_Lemma_BISO}.
\end{proof}

Similar to \cite{Calmon_bounds_Inference}, we also consider the tradeoff between strong estimation privacy and the probability of correctly guessing $Y$. To quantify this, let $\hat{Y}:\Z\to \Y$ be the Bayes decoding map. The resulting (minimum) error probability is $\Pr(\hat{Y}(Z)\neq Y)$. Let
\begin{equation}\label{def:error_probability_EPS}
  \sP(X;Y):=\min_{Z\in\partial\Gamma_{\eps}}\Pr(\hat{Y}(Z)\neq Y).
\end{equation}
Note that when $Z$ is independent of $Y$, then the optimal Bayes decoding map yields $\Pr(\hat{Y}(Z)\neq Y)=1-p$, if $p=P_Y(1)\geq \frac{1}{2}$.  Using a similar argument as \cite[Appendix A]{Shamai_MMSE_Error}, we can establish the following connection between $\sP(X;Y)$ and $\sW_{\eps}(X;Y)$. \begin{proposition}\label{preposition_P_E}
  Let $Y\sim\sBer(p)$ for $p\geq \frac{1}{2}$. Then we have
  $$\sW_{\eps}(X;Y)\leq\frac{\sP(X;Y)}{\var(Y)}\leq 2\sW_{\eps}(X;Y)$$
\end{proposition}
\begin{proof}
First note that
$$\E[Y|Z=z]=P_{Y|Z}(1|z)=\frac{pP_+(z)}{(1-p)P_-(z)+pP_+(z)},$$ where $P_+(z):=P_{Z|Y}(z|+1)$ and $P_-(z):=P_{Z|Y}(z|-1)$. It follows that
\begin{eqnarray*}
  \mmse(Y|Z)&=& \sum_{z\in\Z}\sum_{y\in\{0,1\}}P_{YZ}(y,z)\E[\left(y-\E[Y|Z=z]\right)^2]\\
   &=& p(1-p)\sum_{z\in \Z}\frac{P_-(z)P_+(z)}{(1-p)P_-(z)+pP_+(z)}\\
   &=&p(1-p)[\sum_{z\in\Z_+}\frac{P_-(z)P_+(z)}{(1-p)P_-(z)+pP_+(z)}\\
 &&+\sum_{z\in\Z_-}\frac{P_-(z)P_+(z)}{(1-p)P_-(z)+pP_+(z)}],
\end{eqnarray*}
where $\Z_-=\{z\in\Z:(1-p)P_-(z)\geq pP_+(z)\}$ and $\Z_+=\{z\in\Z:pP_+(z)\geq (1-p)P_-(z)\}$.
Since $$\Pr(\hat{Y}(Z)\neq Y )=p\sum_{z\in\Z_-}P_+(z)+(1-p)\sum_{z\in\Z_+}P_-(z),$$ we then have
$$\frac{1}{2}\Pr(\hat{Y}(Z)\neq Y)\leq\mmse(Y|Z)\leq \Pr(\hat{Y}(Z)\neq Y),$$
from which the result follows immediately.
\end{proof}
Calmon et al.\ \cite{Calmon_bounds_Inference} considered the same problem for $X=Y$, i.e., minimizing $\Pr(\hat{X}(Z)\neq X)$ over all $P_{Z|X}$ such that $\rho_m^2(X;Z)\leq \eps$ and showed that the best privacy-constrained error probability is lower bounded by a straight line of $\eps$ with negative slope. Combining Theorem~\ref{Lemma_BISO} and Proposition~\ref{preposition_P_E}, we can lower bound $\sP(X;Y)$ for all BISO $P_{X|Y}$ by a straight line in $\eps$  as follows:
$$\sP(X;Y)\geq \var(Y)-\eps\frac{\var(X)}{4\E^2[X|Y=1]},$$ which generalizes \cite[Corollaries 3,5]{Calmon_bounds_Inference}.

In the following, we consider two examples of BISO channels for which the bounds in Theorem~\ref{Lemma_BISO} coincide. First consider $P_{X|Y}$ being a binary symmetric channel with crossover probability $\alpha$, denoted as $\text{BSC}(\alpha)$.
\begin{lemma}\label{Lemma_Binary_Y}
For $Y\sim\sBer(p)$ and $P_{X|Y}= \text{BSC}(\alpha)$ for $\alpha\in[0, \frac{1}{2})$, we have for $0\leq\eps\leq \rho_m^2(X;Y)$,
$$1-\frac{\eps\var(X)}{4(1-2\alpha)^2\var(Y)}\leq\sM_{\eps}(X;Y)\leq 1-\frac{\eps}{\rho_m^2(X;Y)},$$
and
$$\var(Y)-\frac{\eps\var(X)}{4(1-2\alpha)^2}\leq\sP(X;Y)\leq 2\left[\var(Y)-\frac{\eps\var(X)}{4(1-2\alpha)^2}\right].$$
Moreover, if $p=\frac{1}{2}$,
$$\sM_{\eps}(X;Y)=\sW_{\eps}(X;Y)=1-\frac{\eps}{(1-2\alpha)^2},$$ and the optimal channel is $\text{BEC}(\tilde{\delta})$ (Fig.~\ref{fig:Bsc}) where
\begin{equation}\label{Delta_Tilde_BSC}
  \tilde{\delta}=1-\frac{\eps}{(1-2\alpha)^2}.
\end{equation}
\end{lemma}
\begin{proof}
  Since $\X=\{-1, +1\}$, it is straightforward to see that $\E[X|Y=1]=1-2\alpha$, and $4\var(Y)(1-2\alpha)^2=\var(X)-4\alpha(1-\alpha)$, and for a fixed $0\leq\alpha< \frac{1}{2}$, $\rho_m^2(X;Y)\leq (1-2\alpha)^2$, which is tight if and only if $p=0.5$. The results follow from Theorem~\ref{Lemma_BISO} and Proposition~\ref{preposition_P_E}. Since for $p=0.5$, the upper bound of Corollary~\ref{Corollary_UpperBound_M_eps} is achieved, hence according to Remark~\ref{remark_Erasure_Optimal_Corollary_2}, the optimal privacy filter is an erasure channel with erasure probability \eqref{Delta_Tilde_BSC}.
\end{proof}
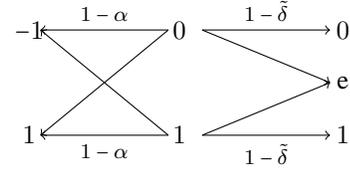
\begin{figure}
\centering
\begin{tikzpicture}
\node (a) [circle] at (0,0) {$1$};
\node (b) [circle] at (0,1.4) {$-1$};
\node (c) [circle] at (2,0) {$1$};
\node (d) [circle] at (2,1.4) {$0$};
\draw[<-] (0.15,0) -- (1.85,0) node[pos=.5,sloped,below] {\footnotesize{$1-\alpha$}};
\draw[<-] (0.15,0) -- (1.85,1.4) node[pos=.75,sloped,below] {};
\draw[<-] (0.15,1.4) -- (1.85,0) node[pos=.25,sloped,above] {};
\draw[<-] (0.15,1.4) -- (1.85,1.4) node[pos=.5,sloped,above] {\footnotesize{$1-\alpha$}};
\node (e) [circle] at (2.3,0) {};
\node (f) [circle] at (2.3,1.4) {};
\node (g) [circle] at (4.17,0) {$1$};
\node (h) [circle] at (4.17,1.4) {$0$};
\node (k) [circle] at  (4.17, 0.7) {\text{e}};
\draw[->] (2.3,0) -- (4,0) node[pos=.5,sloped,below] {\footnotesize{$1-\tilde{\delta}$}};
\draw[->] (2.3,0) -- (4, 0.7) node[pos=.5,sloped,below] {};
\draw[->] (2.3,1.4) -- (4,1.4) node[pos=.5,sloped,above] {\footnotesize{$1-\tilde{\delta}$}};
\draw[->] (2.3,1.4) -- (4, 0.7) node[pos=.5,sloped,above] {};
\end{tikzpicture}
\caption{\small{Optimal privacy filter where $P_{Y|X}=BSC(\alpha)$ with $Y\sim\sBer(\frac{1}{2})$ where $\tilde{\delta}$ is specified in \eqref{Delta_Tilde_BSC}}.} \label{fig:Bsc}
\end{figure}

We next consider $P_{X|Y}$ being a binary erasure channel with erasure probability $\delta$, denoted as $\text{BEC}(\delta)$.
\begin{lemma}\label{Lemma_Binary_Y_BEC}
For $Y\sim\sBer(p)$ and $P_{X|Y}= \text{BEC}(\delta)$ for $\delta\in[0, 1)$, we have for $0\leq\eps\leq \rho_m^2(X;Y)$,
$$1-\frac{\eps\var(X)}{4\var(Y)(1-\delta)^2}\leq\sM_{\eps}(X;Y)\leq1-\frac{\eps}{1-\delta},$$
and
$$\var(Y)-\frac{\eps\var(X)}{4(1-\delta)^2}\leq\sP(X;Y)\leq 2\left[\var(Y)-\frac{\eps\var(X)}{4(1-\delta)^2}\right].$$
Moreover, if $p=\frac{1}{2}$,
$$\sM_{\eps}(X;Y)=1-\frac{\eps}{1-\delta},$$ and the optimal channel is $\text{BEC}(\tilde{\delta})$ (Fig.~\ref{fig:BEC}) where
\begin{equation}\label{Delta_Tilde_BEC}
  \tilde{\delta}=1-\frac{\eps}{1-\delta}.
\end{equation}
\end{lemma}
\begin{proof}
  Since $\X=\{-1, 0, +1\}$, it is easy to show that $\E[X|Y=1]=1-\delta$, and $4(1-\delta)^2\var(Y)=\var(X)-\delta(1-\delta)$, and $\rho_m^2(X;Y)=1-\delta$. When $p=0.5$, then $\var(X)=1-\delta$. Here, again, we see that for uniform $Y$, $\sM_\eps(X;Y)$ achieves the bound given in Corollary~\ref{Corollary_UpperBound_M_eps} and hence again, according to Remark~\ref{remark_Erasure_Optimal_Corollary_2}, the erasure channel is an optimal privacy filter.
\end{proof}

\begin{figure}[t]
\centering
\begin{tikzpicture}
\node (a) [circle] at (0,0) {1};
\node (b) [circle] at (0,1.4) {-1};
\node (e1) [circle] at (0,0.7) {$0$};
\node (c) [circle] at (2.1,0) {1};
\node (d) [circle] at (2.1,1.4) {0};
\draw[->] (1.85,0) -- (0.15,0) node[pos=.5,sloped,below] {\footnotesize{$1-\delta$}};
\draw[->] (1.85,1.4) --(0.15,0.7) node[pos=.75,sloped,below] {};
\draw[<-] (0.15,0.7) -- (1.85,0) node[pos=.25,sloped,above] {};
\draw[<-] (0.15,1.4) -- (1.85,1.4) node[pos=.5,sloped,above] {\footnotesize{$1-\delta$}};
\node (e) [circle] at (2.3,0) {};
\node (f) [circle] at (2.3,1.4) {};
\node (g) [circle] at (4.2,0) {1};
\node (h) [circle] at (4.2,1.4) {-1};
\node (k) [circle] at  (4.2, 0.7) {0};
\draw[->] (2.3,0) -- (4,0) node[pos=.5,sloped,below] {\footnotesize{$1-\tilde{\delta}$}};
\draw[->] (2.3,0) -- (4, 0.7) node[pos=.5,sloped,below] {};
\draw[->] (2.3,1.4) -- (4,1.4) node[pos=.5,sloped,above] {\footnotesize{$1-\tilde{\delta}$}};
\draw[->] (2.3,1.4) -- (4, 0.7) node[pos=.5,sloped,above] {};
\end{tikzpicture}
\caption{\small{Optimal privacy filter where $P_{X|Y}=BEC(\delta)$ with $Y~\sim~\sBer(\frac{1}{2})$ where $\tilde{\delta}$ is specified in \eqref{Delta_Tilde_BEC}}.} \label{fig:BEC}
\end{figure}

We conclude this section by connecting the above results to the \emph{initial efficiency}.
For BISO channels, we define the initial efficiency\footnote{Initial efficiency was previously defined for the common randomness problem in \cite{Lei_Zhai_PhD_Thesis}, for secret key generation in \cite{Key-Generation_product_sources}, for incremental
growth rate in a stock market \cite{erkip}, for source coding problems with side information in \cite{Gohari-additivity}, and for information extraction under privacy constraint in \cite{Asoode_submitted}.} of  $f_\eps(X;Y):=\var(Y)-\var(Y)\sW_{\eps}(X;Y)$ with respect to $\eps$ as  the derivative $f'_{0}(X;Y)$ of $\eps\mapsto f_{\eps}(X;Y)$ at $\eps=0$. In fact, $f'_{0}(X;Y)$ quantifies the decrease of $\mmse(Y|Z)$ when $\eps$ slightly increases from $0$. Then since for any BISO $P_{X|Y}$,  $f_0(X;Y)=0$, using Corollary~\ref{corollary_M_over_eps_Decreasing} and the convexity of $\eps\mapsto\sW_{\eps}(X;Y)$, we can write
\begin{eqnarray*}
  f'_0(X;Y) &=& \lim_{\eps\downarrow 0}\frac{f_{\eps}(X;Y)}{\eps}=\sup_{\eps> 0}\frac{f_{\eps}(X;Y)}{\eps} \\
   &=& \var(X)\max_{P_{Z|Y}:\atop X\markov Y\markov Z}\frac{\var(Y)-\mmse(Y|Z)}{\var(X)-\mmse(X|Z)}.
\end{eqnarray*}
We can, therefore, conclude from Theorem~\ref{Lemma_BISO} that for a given pair of random variables $(X,Y)$ with BISO $P_{X|Y}$, we have
$$\max_{P_{Z|Y}:\atop X\markov Y\markov Z}\frac{\var(Y)-\mmse(Y|Z)}{\var(X)-\mmse(X|Z)}=\frac{1}{4\E^2[X|Y=1]}.$$

\subsection{$\sM_\eps(X;Y)$ and $\sW_{\eps}(X;Y)$ with $n$ {i.i.d.} observations}
Let $(X^n, Y^n)$ be $n$ i.i.d. copies of $(X,Y)$ with a given distribution $P_{XY}$. Similar to \eqref{Def:M_Eps} and \eqref{Def:W_Eps}, we can define
$$\sM_{\eps}(X^n;Y^n):=1-\frac{1}{n}\sup_{Z\in\Gamma_{\eps}^{\otimes n}}\sum_{i=1}^n \eta^2_{Z^n}(Y_i)$$
and
$$\sW_{\eps}(X^n;Y^n):=1-\frac{1}{n}\sup_{Z\in\partial \Gamma_{\eps}^{\otimes n}}\sum_{i=1}^n\eta^2_{Z^n}(Y_i)$$
where $Z^n:=(Z_1, \dots, Z_n)$, and $$\Gamma_{\eps}^{\otimes n}:=\{P_{Z^n|Y^n}: \rho_m^2(X^n; Z^n)\leq \eps\},$$  and $$\partial\Gamma_{\eps}^{\otimes n}:=\{P_{Z^n|Y^n}: \sum_{i=1}^n \eta^2_{Z^n}(X_i)\leq n\eps\}.$$

Using a technique developed in \cite{Calmon_fundamental-Limit}, we can directly show that $\sM_0(X;Y)<1$ if and only if the smallest singular value, $\sigma_{\mathsf{min}}$ , of the operator  $f(X)\mapsto \E[f(X)|Y]$ is zero. Now if we consider the operator $f(X^n)\mapsto \E[f(X^n)|Y^n]$  for i.i.d.\ $(X^n, Y^n)$, we can see that the smallest singular value is $\sigma_{\mathsf{min}}^n$ (see, e.g., \cite{Ulukus_Data_processing}, \cite{polyanskiy2012hypothesis}). It therefore follows that unless $\sigma_{\mathsf{min}}=1$, $\lim_{n\to \infty}\sM_0(X^n;Y^n)<1$ for any distribution $P_{XY}$. This can also be seen from the Slepian-Wolf theorem \cite[Theorem 15.4.1]{Cover_Book} and specifically \cite[Lemma 1]{Tandon_Quantize_bin}.
The following result implies that the optimal privacy filter $P_{Z^n|Y^n}$ which achieves non-trivial perfect privacy  cannot be a memoryless channel.
\begin{proposition}
Let $(X^n, Y^n)$ be an i.i.d. copies of $(X,Y)$ with distribution $P_{XY}$. If the family of feasible stochastic kernels in the optimization \eqref{Eq:PAMMSE_equivalent} is constrained to be of the form $P_{Z^n|Y^n}(z^n|y^n)=\prod_{i=1}^{n}P_{i}(z_i|y_i)$, then
$$\sM_{\eps}(X^n;Y^n)=\sM_{\eps}(X;Y),$$
$$\sW_{\eps}(X^n; Y^n)=\sW_{\eps}(X;Y).$$
\end{proposition}
\begin{proof}
It is clear that $\sM_{\eps}(X^n; Y^n)$ is at most as large as $\sM_{\eps}(X; Y)$, and therefore we will only show $\sM_{\eps}(X^n; Y^n)\geq \sM_{\eps}(X; Y)$ (similarly for $\sW_\eps(X;Y)$).
Let $\eps_i=\rho^2_m(X_i; Z_i)$ for $1\leq i\leq n$. From the tensorization property of maximal correlation \cite{Witsenhausen:dependent}, we know that $\rho_m(X^n; Z^n)=\max\{\rho_m(X_i; Z_i)\}$ and hence $P_{Z^n|Y^n}\in\Gamma_{\eps}^{\otimes n}$ if and only if $\eps_i\leq \eps$ for $1\leq i\leq n$. We can then write
\begin{eqnarray*}
  1-\frac{1}{n}\sum_{i=1}^n\eta^2_{Z^n}(Y_i)&=&\frac{1}{n\var(Y)}\sum_{i=1}^n\mmse(Y_i|Z^n)\\
&=& \frac{1}{n\var(Y)}\sum_{i=1}^n\mmse(Y_i|Z_i)\\
   &\geq&\frac{1}{n}\sum_{i=1}^n\sM_{\eps_i}(X; Y)\\
&\geq& \sM_{\eps}(X;Y),
\end{eqnarray*}
where the last inequality is due to the fact that $\eps\mapsto \sM_{\eps}(X;Y)$ is decreasing. It therefore follows that
$$\sM_{\eps}(X^n;Y^n)\geq \sM_{\eps}(X;Y).$$
To prove the same result for $\sW_{\eps}(X^n; Y^n)$, let now $\eta^2_{Z_i}(X_i)\leq \eps_i$ or equivalently $\mmse(X_i|Z_i)\geq (1-\eps_i)\var(X)$ for $0\leq \eps_i\leq 1$ and $1\leq i\leq n$; hence $P_{Z^n|Y^n}\in\partial\Gamma_{\eps}^{\otimes n}$ if  $\sum_{i=1}^n\eps_i=n\eps$. We can write
\begin{eqnarray*}
   1-\frac{1}{n}\sum_{i=1}^n\eta^2_{Z^n}(Y_i)&=&\frac{1}{n\var(Y)}\sum_{i=1}^n\mmse(Y_i|Z^n)\\&=& \frac{1}{n\var(Y)}\sum_{i=1}^n\mmse(Y_i|Z_i)\\
   &\geq&\frac{1}{n}\sum_{i=1}^n\sW_{\eps_i}(X; Y)\\&\geq& \sW_{\eps}(X;Y),
\end{eqnarray*}
where the last inequality is due to the convexity of $\eps\mapsto \sW_{\eps}(X;Y)$.
\end{proof}
\section{Continuous $(X,Y)$, Additive Gaussian Noise As Privacy Filter}
In this section, we assume $X$ and $Y$ are both absolutely continuous random variables and  the channel $P_{Z|Y}$ is modelled by a scaled additive stable\footnote{A random variable $X$ with distribution $P$ is called stable if for $X_1$, $X_2$ i.i.d.\ according to $P$, for any constants $a$, $b$,
the random variable $aX_1 + bX_2$ has the same distribution as $cX + d$ for some constants $c$ and
$d$ \cite[Chapter 1]{Stable_Nolan}.} noise variable  $N_f$ which is independent of $(X,Y)$ and has density $f$ with zero mean and unit variance, i.e.,
$$Z_{\gamma}=Y+\gamma N_{f},$$ for some $\gamma\geq 0$.  We then define
$$\sM_{\eps}^f(X;Y):=1-\sup_{\gamma\in \C_{\eps}(P_{XY})}\eta^2_{Z_{\gamma}}(Y),$$ and similarly
$$\sW_{\eps}^f(X;Y):=1-\sup_{\gamma\in \partial\C_{\eps}(P_{XY})}\eta^2_{Z_{\gamma}}(Y),$$
where $$\C_{\eps}(P_{XY}):=\{\gamma\geq 0:\rho_m^2(X; Z_{\gamma})\leq \eps\},$$ and $$\partial\C_{\eps}(P_{XY}):=\{\gamma\geq 0:\eta^2_{Z_{\gamma}}(X)\leq\eps\}.$$
If the noise process is Gaussian $N(0,1)$, we denote $N_f$, $\sM_\eps^f(X;Y)$, and $\sW_\eps^f(X;Y)$ by $N_{\sG}$, $\sM_\eps(X;Y)$, and $\sW_\eps(X;Y)$, respectively.

The bounds for $\sW_\eps(X;Y)$ obtained in \eqref{Bounds_UP_LB_Weak} clearly hold:
 $$0\leq\sW_\eps^f(X;Y)\leq\sM^f_\eps(X;Y)\leq 1-\eps,$$
 and, in particular, $\sM_0^f(X;Y)\leq 1$. In the following, we show that this last inequality is in fact an equality.
\begin{proposition}
For a given absolutely continuous $(X,Y)$, the map $\eps\mapsto \sM_\eps^f(X;Y)$ is non-negative, strictly decreasing and satisfies  $$\lim_{\eps\downarrow 0}\sM^f_{\eps}(X;Y)=1.$$
\end{proposition}
\begin{proof}
The proof is similar to the proof of \cite[Theorem 6]{Asoode_submitted} and is hence omitted.
\end{proof}
\begin{example}\label{example_Gaussian}
Let $(X,Y)$ be jointly Gaussian with correlation coefficient $\rho$ and let $N_f=N_{\sG}$. Without loss of generality, we can assume  that $\E[X]=\E[Y]=0$. It is known \cite{Renyi-dependence-measure} that $\rho^2_m(X;Z_{\gamma})=\rho^2(X; Z_{\gamma})$ and hence $$\rho_m^2(X; Z_{\gamma})=\rho^2\frac{\var(Y)}{\var(Y)+\gamma^2},$$ which implies that $\gamma\mapsto \rho_m^2(X; Z_{\gamma})$ is strictly decreasing and hence $\rho_m^2(X; Z_{\gamma})=\eps$ for $0\leq\eps\leq\rho^2_m(X;Y)=\rho^2$ has a unique solution
$$\gamma_{\eps}^2:=\var(Y)\left(\frac{\rho^2}{\eps}-1\right)$$ and $Z_{\gamma}\in \Gamma_{\eps}$ for any $\gamma\geq \gamma_{\eps}$. On the other hand, $\mmse(Y|Z_{\gamma})=\var(Y)\frac{\gamma^2}{\var(Y)+\gamma^2}$ which shows that the map $\gamma\mapsto \mmse(Y|Z_{\gamma})$ is strictly increasing and hence
\begin{equation}\label{M_eps_gaussian}
  \sM_{\eps}(X;Y)=\frac{\mmse(Y|Z_{\gamma_{\eps}})}{\var(Y)}=1-\frac{\eps}{\rho^2}.
\end{equation}
It is easy to check that that $\eta^2_{Z_{\eps}}(X)=\rho_m^2(X;Z_{\eps})=\eps$
This then implies that for the jointly Gaussian $(X, Y)$, $\C_\eps(P_{XY})=\partial\C_\eps(P_{XY})$, i.e.,   the $\eps$-strong estimation privacy \eqref{Def:strong_estimation_privacy} coincides with the $\eps$-weak estimation privacy when $Y$ is perturbed by  Gaussian noise. It then follows that for $0\leq \eps\leq \rho^2$
\begin{equation}\label{Eq:SM_SW_Gaussian}
  \sM_{\eps}(X;Y)=\sW_{\eps}(X;Y)=1-\frac{\eps}{\rho^2}.
\end{equation}
\end{example}

This example suggests that the bound in Corollary~\ref{Corollary_UpperBound_M_eps} still holds for absolutely continuous $(X,Y)$ in this model. We prove this observation in the following lemma with the assumption that $N=N_{\sG}$.
\begin{lemma}\label{Lemma_Gaussian_Input_worst}
For a given absolutely continuous $(X,Y)$, we have for $0\leq \eps\leq \rho_m^2(X;Y)$ $$\sW_{\eps}(X;Y)\leq\sM_{\eps}(X;Y)\leq 1-\frac{\eps}{\rho_m^2(X;Y)}.$$
\end{lemma}
\begin{proof}
It suffices to prove the upper bound as the lower bound follows immediately from \eqref{Bounds_UP_LB_Weak}. Let $\B_{\eps}(P_{XY}):=\{\gamma\geq 0: \rho_m^2(Y;Z_{\gamma})\leq \frac{\eps}{\rho_m^2(X;Y)}\}$. The strong data processing inequality for maximal correlation \cite[Lemma 4]{Asoode_submitted} states that $\rho_m^2(X; Z_{\gamma})\leq \rho_m^2(X;Y)\rho_m^2(Y; Z_{\gamma})$ and therefore
implies $\B_{\eps}(P_{XY})\subseteq \C_{\eps}(P_{XY})$. Therefore
\begin{eqnarray}
  \inf_{\gamma\in \C_{\eps}(P_{XY})}\mmse(Y|Z_{\gamma}) &\leq& \inf_{\gamma\in \B_{\eps}(P_{XY})}\mmse(Y|Z_{\gamma}) \nonumber\\
   &=& \var(Y)\left(1-\frac{\eps}{\rho_m^2(X;Y)}\right)\nonumber,
\end{eqnarray}
where the equality follows form \eqref{Inequality_MMSE_rho}.
\end{proof}
Combined with \eqref{Eq:SM_SW_Gaussian}, this lemma also shows that among all $(X,Y)$ with identical maximal correlation, the jointly Gaussian $(X_{\sG},Y_{\sG})$ yields the largest $\sM_{\eps}(X;Y)$ when the noise process is Gaussian. This observation is similar to \cite[Theorem 12]{MMSE_WU_Properties} which states that for Gaussian noise, the Gaussian input is the worst  with no privacy constraint imposed, i.e., $\mmse(Y|Y+N_{\sG})\leq \mmse(Y_{\sG}|Y_{\sG}+N_{\sG})$ where $Y_{\sG}$ has the same variance as $Y$. Conversely, Wu et al.\ \cite{MMSE_WU_Properties} also showed that for Gaussian input $Y$, additive Gaussian noise is the worst, i.e., $\mmse(Y_{\sG}|Y_{\sG}+N)\leq \mmse(Y_{\sG}|Y_{\sG}+N_{\sG})$ where $N_{\sG}$ is Gaussian having the same variance as $N$. These dual results are essentially the same by switching $Y$ to $N$ because $\mmse(Y|Y+N)=\mmse(N|Y+N)$. However, in our context, the noise variance is the parameter of optimization, and hence the dual of Lemma~\ref{Lemma_Gaussian_Input_worst} is not clear.

We can also obtain a lower bound on $\sM_\eps(X;Y)$ when only $Y$ is Gaussian.

%


\begin{lemma}\label{Lemma_Gaussian_Y_Arbit_X}
Let $X$ be jointly distributed with Gaussian $Y_{\sG}$. Then,
$$1-\frac{\eps}{\rho^2(X; Y_{\sG})}\leq\sM_{\eps}(X; Y_{\sG})\leq 1-\frac{\eps}{\rho_m^2(X; Y_{\sG})},$$
\end{lemma}
\begin{proof}
  First note that
  \begin{eqnarray*}
    \rho_m^2(X; Y_{\sG}+\gamma N_{\sG}) &\geq&\rho^2(X; Y_{\sG}+\gamma N_{\sG})\\
     &=&\rho^2(X; Y_{\sG})\rho^2(Y_{\sG}; Y_{\sG}+\gamma N_{\sG})\\
     &=& \rho^2(X; Y_{\sG})\frac{\var(Y_{\sG})}{\var(Y_{\sG})+\gamma^2}\\
     &=:& \zeta(X; Y_{\sG}).
  \end{eqnarray*}
 Therefore we have
  \begin{equation*}
    \inf_{\gamma\in\C_{\eps}(X; Y_{\sG})} \mmse(Y_{\sG}|Y_{\sG}+\gamma N_{\sG}) \geq  \inf_{\zeta(X; Y_{\sG})\geq \eps}\frac{\gamma^2\var(Y)}{\var(Y)+\gamma^2},
  \end{equation*}
and hence
$$\sM_{\eps}(X; Y_{\sG})\geq 1-\frac{\eps}{\rho^2(X; Y_{\sG})}.$$
\end{proof}
This lemma, together with Example~\ref{example_Gaussian}, implies that
\begin{eqnarray*}
  &&\sM_{\eps}(X_{\sG}, Y_{\sG})-\sM_{\eps}(X;Y_{\sG})\\
 &&\qquad \qquad\qquad \qquad\leq \eps\left[\frac{1}{\rho^2(X; Y_{\sG})}-\frac{1}{\rho^2_m(X; Y_{\sG})}\right]
\end{eqnarray*} for Gaussian $X_{\sG}$ which satisfies $\rho_m^2(X_{\sG}; Y_{\sG})=\rho_m^2(X; Y_{\sG})$. Assume that  the difference $\rho_m^2(X; Y_{\sG})-\rho^2(X; Y_{\sG})$ is small.  Note that this does not necessarily mean that the distribution of $X$ is close to Gaussian. Nevertheless, this lemma illustrates that $\sM_{\eps}(X;Y_{\sG})$ is very close to $\sM_{\eps}(X_{\sG};Y_{\sG})$.

\newcounter{tempequationcounter}
\begin{figure*}[!t]
\normalsize
\setcounter{tempequationcounter}{\value{equation}}
\begin{IEEEeqnarray}{rCl}
\setcounter{equation}{23}
&&\hspace{-0.35cm}\mmse(X|Z)=\E[\var(X|Z)]\nonumber\\
  &=& \sum_{z\in\Z}P_Z(z)\var(X|Z=z)= \sum_{z\in\Z}P_Z(z)\var(P_{X|Z}(\cdot|z))\nonumber\\
   &\stackrel{(a)}{=}&  \sum_{z\in\Z}P_Z(z)\left[\sum_{i=1}^ki^2[P_{X|Z}(i|z)+P_{X|Z}(-i|z)]-\left(\sum_{i=1}^ki[P_{X|Z}(i|z)-P_{X|Z}(-i|z)]\right)^2\right]\nonumber\\
   &\stackrel{(b)}{=}&  \sum_{z\in\Z}P_Z(z)\left[\sum_{i=1}^ki^2[P_{X|Y}(i|1)+P_{X|Y}(-i|1)]-\left(\sum_{i=1}^ki[P_{X|Y}(i|1)-P_{X|Y}(-i|1)][P_{Y|Z}(1|z)-P_{Y|Z}(0|z)]\right)^2\right]\nonumber\\
   &=&\sum_{z\in\Z}P_Z(z)\left[\sum_{i=1}^ki^2[P_{X|Y}(i|1)+P_{X|Y}(-i|1)]-[P_{Y|Z}(1|z)-P_{Y|Z}(0|z)]^2\left(\sum_{i=1}^ki[P_{X|Y}(i|1)-P_{X|Y}(-i|1)]\right)^2\right]\nonumber\\
&=&\sum_{z\in\Z}P_Z(z)\left[\sum_{i=1}^ki^2[P_{X|Y}(i|1)+P_{X|Y}(-i|1)]-[2\var^{-1}_b(\var_b(Y|Z=z))-1]^2\left(\sum_{i=1}^ki[P_{X|Y}(i|1)-P_{X|Y}(-i|1)]\right)^2\right]\nonumber\\
&\stackrel{(c)}{=}&\sum_{z\in\Z}P_Z(z)\left[\sum_{i=1}^ki^2[P_{X|Y}(i|1)+P_{X|Y}(-i|1)]-[1-4\var_b(Y|Z=z)]\left(\sum_{i=1}^ki[P_{X|Y}(i|1)-P_{X|Y}(-i|1)]\right)^2\right]\nonumber\\
&=&\sum_{i=1}^ki^2[P_{X|Y}(i|1)+P_{X|Y}(-i|1)]-\left(\sum_{i=1}^ki[P_{X|Y}(i|1)-P_{X|Y}(-i|1)]\right)^2\left(1-4\E[\var(Y|Z)]\right),
\label{eq:floatingequation_BISO}
\end{IEEEeqnarray}
\setcounter{equation}{\value{tempequationcounter}}
\hrulefill
\end{figure*}


\begin{figure*}[t]
\normalsize
\setcounter{tempequationcounter}{\value{equation}}
\begin{IEEEeqnarray}{rCl}
\setcounter{equation}{24}
   \mmse(Y|Z)&=& \frac{\left(\sum_{i=1}^ki[P_{X|Y}(i|1)-P_{X|Y}(-i|1)]\right)^2-\sum_{i=1}^ki^2[P_{X|Y}(i|1)+P_{X|Y}(-i|1)]+\mmse(X|Z)}{4\left(\sum_{i=1}^ki[P_{X|Y}(i|1)-P_{X|Y}(-i|1)]\right)^2}\nonumber\\
  &=& \frac{\mmse(X|Z)-\var(X|Y=1)}{4\left(\sum_{i=1}^ki[P_{X|Y}(i|1)-P_{X|Y}(-i|1)]\right)^2}=\frac{\mmse(X|Z)-\var(X|Y=1)}{4\E^2[X|Y=1]}. \label{eq:floatingequation_BISO_2}
\end{IEEEeqnarray}
\hrulefill
\end{figure*}

\appendices
\section{Proof of Theorem~\ref{Lemma_BISO}}\label{Appebdix_Lemma_BISO}
For $Y\sim\sBer(p)$, we have $\var_b(Y):=\var(P_Y)=p(1-p)$ and let $\var_b^{-1}:[0, \frac{1}{4}]\to [0, \frac{1}{2}]$ be its inverse function. Due to the Markovity condition $X\markov Y\markov Z$, we can write
\begin{equation}\label{markovity}
  P_{X|Z}(x|z)=P_{X|Y}(x|1)P_{Y|Z}(1|z)+P_{X|Y}(x|0)P_{Y|Z}(0|z).
\end{equation}
Note that for $X$ supported over $\X=\{\pm 1, \pm 2, \dots, \pm k\}$, the variance can be written as
\begin{equation}\label{variance_BISO}
  \var(X)=\sum_{i=1}^ki^2[P_X(i)+P_X(-i)]-\left[\sum_{i=1}^ki[P_X(i)-P_X(-i)]\right]^2.
\end{equation}

We can expand $\mmse(X|Z)$ as in \eqref{eq:floatingequation_BISO} where $(a)$ is a simple application of \eqref{variance_BISO}, $(b)$ follows from the Markovity condition \eqref{markovity} and the definition of BISO, and in $(c)$ we used the fact that $\var^{-1}_b(u)=\frac{1}{2}(1-\sqrt{1-4u})$ for any $0\leq u\leq \frac{1}{4}$.
\addtocounter{equation}{-1}

We can therefore write $\mmse(Y|Z)$ \emph{linearly} in terms of $\mmse(X|Z)$ as in  \eqref{eq:floatingequation_BISO_2}. Note that since for $Z\in \partial\Gamma_{\eps}$, $\mmse(X|Z)\geq (1-\eps)\var(X)$, we can  write
\addtocounter{equation}{1}
 \begin{eqnarray}
  \sW_{\eps}(X;Y) &=& \frac{(1-\eps)\var(X)-\var(X|Y=1)}{4(\E[X|Y=1])^2} \nonumber\\
   &=&  \frac{\var(X)-\var(X|Y=1)}{4\var(Y)\E^2[X|Y=1]}\nonumber\\
   &&-\frac{\eps\var(X)}{4\var(Y)\E^2[X|Y=1]}\label{variance_BISO_2}.
\end{eqnarray}

Note that, we have
\begin{eqnarray*}
   p\var(X|Y=1)&+&(1-p)\var(X|Y=0)= \E[\var(X|Y)] \\
   &=& \var(X)-\var(\E[X|Y]),
\end{eqnarray*}
and consequently,
\begin{eqnarray}\label{variance_BISO_3}
  &&\var(X)-\var(X|Y=1)=\var(\E[X|Y])\nonumber\\
  &&\qquad\qquad+(1-p)[\var(X|Y=0)-\var(X|Y=1)]\nonumber\\
  &&~~~~~\qquad\qquad\qquad\qquad\stackrel{(a)}{=}\var(\E[X|Y])
\end{eqnarray}
where $(a)$ follows from the symmetry of the channel $P_{X|Y}$. Note that $\E[X|Y]$ is a binary random variable which is equal to $\E[X|Y=1]$  with probability $p$ and $\E[X|Y=0]$  with probability $1-p$. Due to the symmetry of the channel, one can easily show that $\E[X|Y=0]=-\E[X|Y=1]$. It then follows that \begin{eqnarray}
  \var(\E[X|Y])&=&  p(\E[X|Y=1])^2+(1-p)(\E[X|Y=0])^2\nonumber\\
  &&-\left[p\E[X|Y=1]+(1-p)\E[X|Y=0]\right]^2\nonumber\\
   &=& (\E[X|Y=1])^2-(\E[X|Y=1])^2(2p-1)^2\nonumber\\
&=&4p(1-p) (\E[X|Y=1])^2\nonumber\\
&=&4\var(Y)(\E[X|Y=1])^2\label{variance_BISO_4}
\end{eqnarray}
Plugging \eqref{variance_BISO_3} and \eqref{variance_BISO_4} into \eqref{variance_BISO_2}, we can conclude that
\begin{equation}\label{W_eps_BISO}
\sW_{\eps}(X;Y)=1-\frac{\eps\var(X)}{4\var(Y)\E^2[X|Y=1]}.
\end{equation}
The bound for $\sM_\eps(X;Y)$ simple follows from \eqref{W_eps_BISO} and Corollary~\ref{Corollary_UpperBound_M_eps}.

\bibliographystyle{plain}
\bibliography{bibliography}

\begin{thebibliography}{10}

\bibitem{Fady_Lecture_note}
F.~Alajaji and P.~N. Chen.
\newblock {\em Information Theory for Single User Systems, Part {I}}.
\newblock Course Notes, Queen's University,
  \url{http://www.mast.queensu.ca/~math474/it-lecture-notes.pdf}, 2015.

\bibitem{Asoodeh_Allerton}
S.~Asoodeh, F.~Alajaji, and T.~Linder.
\newblock Notes on information-theoretic privacy.
\newblock In {\em Proc.\ 52nd Annual Allerton Conference on Communication,
  Control, and Computing}, pages 1272--1278, Sept. 2014.

\bibitem{Asoodeh_CWIT}
S.~Asoodeh, F.~Alajaji, and T.~Linder.
\newblock On maximal correlation, mutual information and data privacy.
\newblock In {\em Proc.\ IEEE 14th Canadian Workshop on Inf. Theory (CWIT)},
  pages 27--31, June 2015.

\bibitem{Asoode_submitted}
S.~Asoodeh, M.~Diaz, F.~Alajaji, and T.~Linder.
\newblock Information extraction under privacy constraints.
\newblock {\em \url{arXiv:1511.02381}}, 2015.

\bibitem{Gohari-additivity}
S.~Beigi and A.~Gohari.
\newblock On the duality of additivity and tensorization.
\newblock {\em \url{arXiv:1502.00827v1}}, 2015.

\bibitem{Calmon_fundamental-Limit}
F.~P. Calmon, A.~Makhdoumi, and M.~M\'{e}dard.
\newblock Fundamental limits of perfect privacy.
\newblock In {\em Proc.\ IEEE Int. Symp. Inf. Theory (ISIT)}, pages 1796--1800,
  2015.

\bibitem{Calmon_bounds_Inference}
F.~P. Calmon, M.~Varia, M.~M\'{e}dard, M.~M. Christiansen, K.~R. Duffy, and
  S.~Tessaro.
\newblock Bounds on inference.
\newblock In {\em Proc.\ 51st Annual Allerton Conference on Communication,
  Control, and Computing}, pages 567--574, Oct 2013.

\bibitem{Shamai_MMSE_Error}
N.~Chayat and S.~Shamai.
\newblock Bounds on the capacity of a binary input {AWGN} channel with
  intertransition duration restrictions.
\newblock In {\em Proc. 17th Convention of Electrical and Electronics Engineers
  in Israel,}, pages 227--229, March 1991.

\bibitem{Cover_Book}
T.~M. Cover and J.~A. Thomas.
\newblock {\em Elements of Information Theory}.
\newblock Wiley-Interscience, 2006.

\bibitem{csiszarbook}
I.~Csisz\'{a}r and J.~K\"{o}rner.
\newblock {\em Information Theory: Coding Theorems for Discrete Memoryless
  Systems}.
\newblock Cambridge University Press, 2011.

\bibitem{erkip}
E.~Erkip and T.M. Cover.
\newblock "the efficiency of investment information".
\newblock {\em IEEE Trans. Inf. Theory}, 44(3):1026--1040, May 1998.

\bibitem{Shamai_BISO1}
Y.~Geng, C.~Nair, S.~Shamai, and Z.~V. Wang.
\newblock On broadcast channels with binary inputs and symmetric outputs.
\newblock {\em IEEE Trans. Inf. Theory}, 59(11):6980--6989, March 2013.

\bibitem{Goldwasser1984270}
S.~Goldwasser and S.~Micali.
\newblock Probabilistic encryption.
\newblock {\em Journal of Computer and System Sciences}, 28(2):270 -- 299,
  1984.

\bibitem{Ulukus_Data_processing}
W.~Kang and S.~Ulukus.
\newblock A new data processing inequality and its applications in distributed
  source and channel coding.
\newblock {\em IEEE Trans. Inf. Theory}, 57(1):56--69, Jan. 2011.

\bibitem{Key-Generation_product_sources}
J.~Liu, P.~Cuff, and S.~Verd\'{u}.
\newblock Key capacity for product sources with application to stationary
  {Gaussian} processes.
\newblock {\em \url{ arXiv:1409.5844}}, 2014.

\bibitem{Fawaz_Makhdoumi}
A.~Makhdoumi and N.~Fawaz.
\newblock Privacy-utility tradeoff under statistical uncertainty.
\newblock In {\em Proc.\ 51st Allerton Conference on Communication, Control,
  and Computing}, pages 1627--1634, Oct 2013.

\bibitem{Funnel}
A.~Makhdoumi, S.~Salamatian, N.~Fawaz, and M.~M\'{e}dard.
\newblock From the information bottleneck to the privacy funnel.
\newblock In {\em Proc.\ IEEE Inf. Theory Workshop (ITW)}, pages 501--505,
  2014.

\bibitem{Stable_Nolan}
J.~P. Nolan.
\newblock {\em Stable Distributions-Models for Heavy Tailed Data}.
\newblock Boston: Birkhauser, in progress, Chapter 1 online at,
  \url{academic2.american.edu/~jpnolan}, 2010.

\bibitem{polyanskiy2012hypothesis}
Y.~Polyanskiy.
\newblock Hypothesis testing via a comparator.
\newblock In {\em Proc.\ IEEE Int. Symp. Inf. (ISIT)}, pages 2206--2210, July
  2012.

\bibitem{t_closeness}
D.~Rebollo-Monedero, J.~Forne, and J.~Domingo-Ferrer.
\newblock From t-closeness-like privacy to postrandomization via information
  theory.
\newblock {\em IEEE Trans. Knowl. Data Eng.}, 22(11):1623--1636, Nov 2010.

\bibitem{Renyi-dependence-measure}
A.~R\'{e}nyi.
\newblock On measures of dependence.
\newblock {\em Acta Mathematica Academiae Scientiarum Hungarica},
  10(3):441--451, 1959.

\bibitem{Lalitha_Forensics}
L.~Sankar, S.R. Rajagopalan, and H.V. Poor.
\newblock Utility-privacy tradeoffs in databases: An information-theoretic
  approach.
\newblock {\em IEEE Trans. Inf. Forensics Security}, 8(6):838--852, 2013.

\bibitem{Sarmanov}
O.V. Sarmanov.
\newblock The maximum correlation coefficient (nonsymmetric case).
\newblock {\em Dokl. Akad. Nauk SSSR}, 120(4):715--718, 1958.

\bibitem{Shamai_BISO2}
I.~Sutskover, S.~Shamai, and J.~Ziv.
\newblock Extremes of information combining.
\newblock {\em IEEE Trans. Inf. Theory}, 51(4):1313--1325, April 2005.

\bibitem{Tandon_Quantize_bin}
R.~Tandon, L.~Sankar, and H.V. Poor.
\newblock Discriminatory lossy source coding: side information privacy.
\newblock {\em IEEE Trans. Inf. Theory}, 59(9):5665--5677, April 2013.

\bibitem{Counterexample_WISE}
G.~L. Wise.
\newblock A note on a common misconception in estimation.
\newblock {\em Systems and Control Letters}, 5(5):355--356, 1985.

\bibitem{Witsenhausen:dependent}
H.~S. Witsenhausen.
\newblock On sequence of pairs of dependent random variables.
\newblock {\em SIAM Journal on Applied Mathematics}, 28(2):100--113, 1975.

\bibitem{MMSE_WU_Properties}
Y.~Wu and S.~Verd\'{u}.
\newblock Functional properties of minimum mean-square error and mutual
  information.
\newblock {\em IEEE Trans. Inf. Theory,}, 58(3):1289--1301, March 2012.

\bibitem{yamamotoequivocationdistortion}
H.~Yamamoto.
\newblock A source coding problem for sources with additional outputs to keep
  secret from the receiver or wiretappers.
\newblock {\em IEEE Trans. Inf. Theory}, 29(6):918--923, Nov. 1983.

\bibitem{Lei_Zhai_PhD_Thesis}
L.~Zhao.
\newblock {\em Common randomness, efficiency, and actions}.
\newblock PhD thesis, Stanford University, 2011.

\end{thebibliography}

\end{document}